\documentclass[11pt,a4paper]{amsart}

\usepackage[latin1]{inputenc}
\usepackage[english]{babel}
\usepackage{amsmath}
\usepackage{bm}
\usepackage{amssymb, mathabx}
\usepackage{graphicx}
\usepackage{braket}
\usepackage[pdftex,plainpages=false,colorlinks,hyperindex,bookmarksopen,linkcolor=red,citecolor=blue,urlcolor=blue]{hyperref}
\usepackage{cite}
\usepackage{mathrsfs}
\usepackage{epstopdf}
\usepackage{braket}
\usepackage[hmargin=3cm,vmargin={3.5cm,4cm}]{geometry}

\theoremstyle{theorem}
\newtheorem{thm}{Theorem}

\theoremstyle{definition}                                 

\theoremstyle{definition}                           

\theoremstyle{remark}                             

\usepackage{color}

\usepackage{mathtools,slashed}

\newcommand{\be}{\begin{eqnarray}}
\newcommand{\ee}{\end{eqnarray}}

\newcommand{\R}{\mathbb{R}}  
\newcommand{\C}{\mathbb{C}} 
\newcommand{\N}{\mathbb{N}} 
\def\eg{{\it e.g. }} 
\def\ie{{\it i.e. }}

\newcommand{\wt}[1]{\widetilde{#1}}

\newcommand{\ceil}[1]{\lceil #1 \rceil}

\def\d{{\rm d}}

\def\eg{{\it e.g.}\ }
\def\ie{{\it i.e.}\ }

\allowdisplaybreaks
\begin{document}

\title{General fractional calculus and Prabhakar's theory}
	
	    \author{Andrea Giusti}
		\address{ 
		Bishop's University,
		Physics $\&$ Astronomy Department, 
		2600 College Street, Sherbrooke, J1M 1Z7,
		QC	Canada}	
 		\email{agiusti@ubishops.ca}

	
    \date{November 15, 2019}

\begin{abstract}
General fractional calculus offers an elegant and self-consistent path toward the generalization of fractional calculus to an enhanced class of kernels. Prabhakar's theory can be thought of, to some extent, as an explicit realization of this scheme achieved by merging the Prabhakar (or, three-parameter Mittag-Leffler) function with the general wisdom of the standard (Riemann-Liouville and Caputo) formulation of fractional calculus. Here I discuss some implications that emerge when attempting to frame Prabhakar's theory within the program of general fractional calculus.
\end{abstract}

\thanks{\textbf{Comm.~Nonlinear~Sci.~Numer.~Simulat. 83 (2020) 105114}, DOI: \href{https://doi.org/10.1016/j.cnsns.2019.105114}{10.1016/j.cnsns.2019.105114}}

    \maketitle

\section{Introduction} \label{sec:intro}
	In the last few years a long-standing discussion on the very notion of fractional derivative has been brought back into the spotlight (see \eg \cite{OM-1, Tarasov-1, Tarasov-2, Giusti-NODY, Garrappa-Bongo, Stynes, Luchko-Hilfer}). Traditionally, fractional calculus has always shown a strict bond with the theory of {\em singular} (Fredholm-)Volterra integro-differential equations. For this reason I believe that such a feature should be preserved in all alternative definitions of fractional derivative and I tend to think that, once an unambiguous and unique notion of fractional derivative will be formulated, this singular behavior will most likely represent a pivotal point of the theory. All that said, the aim of this work is neither to condemn nor to condone alternative views on this matter, nonetheless I am convinced that being bluntly clear on this point can ease the reader in understanding the origin and motivations of the arguments discussed here. 

	The scope of this study is to discuss the very nature of the Prabhakar derivative by framing it in Kochubei's general fractional calculus. This analysis will show that, despite the great generality of its definition, the Prabhakar derivative inherit the same weakly singular behavior of the traditional fractional derivatives (Riemann-Liouville and Caputo), thus consolidating its position within the realm of fractional derivatives.

	This work is organized as follows: first, in Section \ref{sec:vs}, I recall some of the basic traits of Volterra operators and fractional derivatives; in Section \ref{sec:kebab}, I first recall the main features of the Prabhakar fractional calculus, then I discuss the nature of the Prabhakar kernel; in Section \ref{sec:gfc}, I analyze to what extent one can frame Prabhakar's theory within the framework of Kochubei's general fractional calculus, taking profit of a reverse engineering approach; then, in Section \ref{sec:conclusions}, I provide some concluding remarks.

\section{Fractional Derivatives vs. Volterra operators} \label{sec:vs}

	Let us consider the class of {\em causal functions}, \ie $u(t) = 0$ for $t<0$. A (linear) Volterra operators of the convolution type can be written as (see \eg \cite{Gohberg})
	\be \label{volterra}
	\mathcal{V}_{(k)} u(t) = \int _{0} ^{t} k(t - \tau) \, u(\tau) \, \d \tau
	\equiv (k \star u) (t) \, , 
	\ee
	where $\star$ denotes the convolution product in the Laplace sense. 

	As one can easily infer, Eq. \eqref{volterra} plays a key role in fractional calculus. Indeed, following the old-fashioned approach pioneered by Riemann and Liouville \cite{Mainardi-Gorenflo-1997, Mainardi_BOOK10, SKM}, which still represents the gold standard for a mathematically consistent formulation of the theory, one has that the simplest way to introduce the notion of fractional integral is through a minimal (though non-trivial) analytic continuation of Cauchy's formula for repeated integration. Specifically, let $f \in L^1 _{\rm loc} (\R ^+)$, then the $n$th integral of $f(t)$ reduces to
	\be \label{Cauchy}
	\notag
	J ^{n} f(t) &\!:=\!& 
	\int _0 ^t \d \tau _1 \int _0 ^{\tau _1} \d \tau _2 
	\cdots 
	\int _0 ^{\tau _{n-1}} \d \tau _n \, f (\tau_n) =\\
	&\!=\!& \frac{1}{(n-1)!} \int _0 ^t (t - \tau)^{n-1} \, f(\tau) \, \d \tau \, .
	\ee
Recalling that $(n-1)! \equiv \Gamma (n)$ and that Euler's gamma function $\Gamma (z)$ can be represented in terms of an absolutely convergent integral for $\Re (z) > 0$, one can generalize the previous expression. Specifically, this procedure leads to the so called {\em Riemann-Liouville fractional integral}
	\be \label{RL}
	J ^{\alpha} f(t) := (\Phi _\alpha \star f) (t) \, , \qquad t > 0 \, , \,\, \alpha > 0 \, ,  
	\ee
	where $\Phi _\alpha (t)$ denotes the Gel'fand-Shilov distribution, which is such that 
	$$\Phi _\alpha (t) := \frac{t^{\alpha - 1}}{\Gamma (\alpha)} \, ,$$ 
	for $t>0$ and vanishes otherwise. Hence, from \eqref{RL} one infers that the Riemann-Liouville fractional integral is nothing but a Volterra operator equipped with a Gel'fand-Shilov kernel, which turns out to be {\em weakly singular} if $0<\alpha <1$.
 
All that being said, one can introduce the notion of {\em fractional derivative} starting from \eqref{RL} and mimicking the key feature of a ordinary derivatives, namely the fact that they act as the ``left-inverse'' of the ordinary integral. Thus, the minimal requirements for defining a fractional derivative $D^\alpha$ from the Riemann-Liouville integral are: (i) one has to recover ordinary calculus in the formal limit $\alpha \in \R^+ \to \alpha \in \N \cup \{0\}$; (ii) $D^\alpha \, J^\alpha f(t) = f(t)$ for $f \in L^1 _{\rm loc} (\R ^+)$. Of course, this implies that the {\em uniqueness} of the definition of fractional derivative is {\em inevitably lost}. Indeed, it is not hard to see that both
	\be \label{RLC}
	{}^{\rm RL} D^\alpha f(t) := D^m J^{m-\alpha} f(t) \quad \mbox{and} \quad 
	{}^{\rm C} D^\alpha f(t) := J^{m-\alpha} D^m f(t) \, ,
	\ee
	with $m = \ceil{\alpha}$ and $D^m g(t) \equiv g^{(m)} (t)$ denoting the $m$th derivative of $g(t)$, satisfy the conditions mentioned above. The two operators in \eqref{RLC} are respectively known in the literature as the Riemann-Liouville and Caputo-Dzhrbashyan fractional derivatives. Note that, from \eqref{RL} and \eqref{RLC} one can easily conclude that fractional calculus can be thought of as a subclass of the theory of linear Volterra integro-differential operators. The specific distinction between these two frameworks is however justified by the need for a {\em minimal} extension of ordinary calculus to an arbitrary positive real order. Furthermore, \eqref{RLC} clearly highlights the {\em weakly singular} nature of these fractional derivatives, for all $\alpha > 0$. It is, however, worth stressing that the importance of the weakly singular nature of these operators is even more subtle than the simple arguments presented thus far. Indeed, this weakly singular nature lies deep at the heart of the distributional aspects of fractional calculus. To clarify this statement it is worth recalling a lesser known representation of the $n$th derivative of the Dirac delta distribution \cite{Mainardi-Gorenflo-1997, distributions}, \ie 
$$
\Phi _{-n} (t) := \lim _{\alpha \to n} \Phi _{-\alpha} (t) = \delta ^{(n)} (t) \, , \qquad n \in \N \cup \{0\} \, .
$$
Given this result, it is not hard to see that the $n$th derivative of a causal function $f(t)$ can be obtained as 
$$
D^n f(t) = (\Phi _{-n} \star f) (t) \, , 
$$
with $\star$ denoting an appropriate extension of the Laplace convolution to the realm of distributions. It is therefore easy to convince oneself to think of the convolution
\be \label{dist-frac}
D^\alpha f(t) = (\Phi _{-\alpha} \star f) (t)
\ee
as a well suited definition of fractional derivative at the level of generalized functions. This is undoubtedly a correct, though the whole argument hides some shortcomings. Indeed, the singular behavior of $\Phi _{-\alpha} (t)$ makes the requirements on the regularity of $f(t)$ in Eq. \eqref{dist-frac} too restrictive for most practical purposes, since $\Phi _{-\alpha} (t) \notin L^1 _{\rm loc} (\R^+)$. Hence, one is generally required to regularize the divergent integral coming from the convolution. The simplest thing that can be done consist in healing the kernel by ``weakening its singularity''. This is done by replacing $\Phi _{-\alpha}$ in  \eqref{dist-frac} with $\Phi _{m-\alpha} \star\Phi _{-m}$ or $\Phi _{-m} \star \Phi _{m-\alpha}$. This procedure renders the integral kernel weakly singular and it allows one to recover the results in Eq. \eqref{RLC} (see \eg \cite{Mainardi-Gorenflo-1997,SKM} for further details). 

	A very general and rigorous study of all Volterra-like operators that generalize the traditional construct of fractional calculus has been carried out by Kochubei in \cite{Kochubei}. The underling idea is that, assuming $k \in L^1 _{\rm loc} (\R ^+)$, one can always define a generalized fractional integro-differential operator
	\be \label{C-K}
	\mathcal{D} _{(k)} u(t) := \frac{\d}{\d t} \int _0 ^t k(t - \tau) \, u(\tau) \, \d \tau - k(t) \, u(0^+) \, , 
	\qquad \mbox{with} \,\,\, u \in L^1 _{\rm loc} (\R ^+) \, ,
	\ee  
thus extending the standard scheme presented above. Note that we have decided to focus solely on the regularized (Caputo-like) version of the operator $\mathcal{D} _{(k)}$ for the sake of simplicity and also because it entails a wider scientific interest when dealing with physically relevant initial value problems. Now, under some general assumptions (see \cite{Kochubei} for further details), one can find the generalized fractional integration operator associated to \eqref{C-K}, \ie
	\be \label{I-K}
	\mathcal{J} _{(k)} u(t) = \int _0 ^t \varkappa (t- \tau) \, u(\tau) \, \d \tau \, , 
	\ee
by following the prescriptions for the Riemann-Liouville case. The key feature of $\varkappa (t)$ is that it must satisfy the condition $(k \star \varkappa) (t) = 1$, \ie $k(t)$ and $\varkappa (t)$ form {\em Sonine pair} \cite{Kochubei, samko2003integral, Sonine}. Along this line, it was argued that if two locally integrable functions satisfy the Sonine condition, together with some monotonicity conditions near $t=0$, then they must have an integrable singularity at $t = 0$ \cite{SC-1,SC-2,Hanyga}. This further strengthen the connection between fractional derivatives and weakly singular Volterra operators discussed above.

To sum up, this analysis of the literature on fractional calculus seems to suggest that {\em fractional derivatives represent a class of weakly singular Volterra convolution operators} and that this specific feature should be carried along in all generalization of such objects.

\section{On the weakly singular nature of the Prabhakar kernel}\label{sec:kebab}

	Let me recall a few basic definitions of Prabhakar calculus. First, this whole research topic is based on a three parameter generalization of the Mittag-Leffler function, known as the {\em Prabhakar function} \cite{Prabhakar1971, FM-ML, Garrappa-SIAM,GarrappaPopolizio-JSC}, \ie
\be \label{P-function}
E_{\alpha,\beta}^{\gamma}(z) 
= 
\sum_{k=0}^{\infty} \frac{(\gamma)_k \,  z^{k}}{k! \Gamma(\alpha k + \beta)}
, \quad 
\alpha, \beta, \gamma \in \mathbb{C}
, \quad \Re(\alpha) > 0 \, ,
\ee
where $(\gamma)_k$ is the Pochhammer (rising factorial) symbol which can also be rewritten as $(\gamma)_k = \Gamma (\gamma + k) / \Gamma (\gamma)$. The latter turns out to be an {\em entire function} of order $\rho = 1/\Re(\alpha)$ and type $\sigma=1$ (see \eg \cite{FM-ML}), this will indeed be a key point in the following. 

	The Prabhakar kernel \cite{PBK-derivative} is then defined as
\be \label{P-kernel}
e^{\gamma} _{\alpha, \beta}(\lambda; t) = t^{\beta-1}E^{\gamma}_{\alpha,\beta} \left(\lambda t^{\alpha} \right),
\qquad t \in \R ^+, \: \alpha, \beta, \lambda, \gamma \in \mathbb{C}, \: \Re(\alpha),\Re(\beta)>0.
\ee
Following the principles of the theory of {\em general fractional calculus} \cite{Kochubei}, one can define a generalized fractional integral, and the corresponding derivatives, out of the Prabhakar kernel. In detail, let $f \in L^1 (\R ^+)$, then the Prabhakar integral reads \cite{Prabhakar1971,KSS,PBK-derivative}
\be \label{P-integral}
\bm{\mathcal{E}}^{\gamma}_{\alpha,\beta, \lambda}  f(t) 
= (e^{\gamma}_{\alpha,\beta}(\lambda ; \cdot ) \star f) (t) \, , \quad 
\alpha, \beta, \lambda, \gamma \in \mathbb{C}, \: \Re(\alpha),\Re(\beta)>0 \, .
\ee
Then, following the prescriptions discussed above, one can define the regularized Prabhakar derivative \cite{DOP,PBK-derivative} according to
\be \label{P-derivative}
{}^{\rm C}\bm{\mathcal{D}}^{\gamma}_{\alpha, \beta, \lambda} f(t)&=\bm{\mathcal{E}}^{-\gamma}_{\alpha, m-\beta, \lambda} (D^m f) (t) \, , \qquad m:= \ceil{\beta}
\ee
with $f(t)$ denoting real-valued function whose derivatives are continuous up to order $m-1$ on $\R^+$ and such that $D^{m-1} f$ is absolutely continuous. For further details see \eg \cite{Giusti-NODY, PBK-derivative, GG, MG, GC, Garrappa-HN, CGV-Mathematics2018, TomovskiHilferSrivastava2010} and Appendix \ref{sec:appendix}.

	Note that I am solely focusing on $\R^+ = (0, \infty)$ and the regularized version of the Prabhakar derivative, nonetheless the following arguments can be easily extended beyond these restrictions.

	From Eq. \eqref{P-derivative}, it is now easy to identify the integral kernel that defines this operator, namely
\be \label{kernel-derivative}
e^{- \gamma} _{\alpha, m-\beta}(\lambda; t) = t^{m - \beta -1} \, 
E^{-\gamma}_{\alpha,m - \beta} \left(\lambda t^{\alpha} \right) \, ,
\ee
again with $\alpha, \beta, \lambda, \gamma \in \mathbb{C}$, $\Re(\alpha)>0$, and $\Re(\beta)>0$. From Eq. \eqref{kernel-derivative} it is easy to infer that $E^{-\gamma}_{\alpha,m - \beta} \left(\lambda t^{\alpha} \right)$ is always regular on $\R ^+$, whereas $t^{m - \beta -1}$ gives rise of a power-law singularity in $t=0$, for $\beta \notin \N$, of the same type of the standard fractional derivatives in Eq. \eqref{RLC}. This is in full agreement with the discussion in \cite{Hanyga} and brings all potential suspicious  concerning the fractional nature of Prabhakar derivative to a close. Besides, if one tries to investigate the limit $\epsilon \to 0$, with $\epsilon = m - \beta$, it is easy to see that
\be \label{series-k}
\nonumber e^{- \gamma} _{\alpha,\epsilon}(\lambda; t) &\equiv& e^{- \gamma} _{\alpha, m-\beta}(\lambda; t) \\
\nonumber &=& \frac{t^{\epsilon - 1}}{\Gamma (\epsilon)} + \sum _{k=1} ^\infty 
\frac{(-\gamma)_k \, \lambda ^k }{k!}
\frac{t^{\alpha k + \epsilon - 1}}{\Gamma (\alpha k + \epsilon)} \\
\nonumber &=& \Phi_{\epsilon} (t) + \sum _{k=1} ^\infty 
\frac{(-\gamma)_k \, \lambda ^k }{k!}
\Phi_{\alpha k + \epsilon} (t) \\
&\overset{\epsilon \to 0}{\longrightarrow}& \delta (t) + 
\sum _{k=1} ^\infty 
\frac{(-\gamma)_k \, \lambda ^k }{k!}
\, \Phi _{\alpha k} (t) \, ,
\ee
which implies
\be \label{series-d}
\lim _{\beta \to n}{}^{\rm C}\bm{\mathcal{D}}^{\gamma}_{\alpha, \beta, \lambda} f(t)
= D^n f(t) +  \sum _{k=1} ^\infty 
\frac{(-\gamma)_k \, \lambda ^k }{k!}
(J ^{\alpha k} D^n f) (t) \, , \quad n \in \N \, .
\ee
For details on series representations in Prabhakar calculus I invite the reader to refer to \cite{Giusti-NODY}, and to \cite{scemolo} for some extensions. 

From Eqs. \eqref{series-k} and \eqref{series-d} one can conclude that ordinary calculus is obtained via the double limits: a) $\{\beta \to n \, ; \, \gamma \to 0\}$; b) $\{\beta \to n \, ; \, \lambda \to 0\}$. In both cases the two limits commute. Furthermore, focusing just on $\beta \to n$, one can infer that if $0 < \Re (\alpha) < 1$ the series in the right hand side of Eq. \eqref{series-k} contains at least one singular term at $t=0$, whereas for $\Re (\alpha) \geq 1$ all terms in the series are non-singular at $t=0$. In other words, if $\Re (\alpha) \geq 1$ and $\beta \in \N$ then the Prabhakar derivatives reduces to an ordinary derivative and a series of non-singular, linear Volterra operators which are arguably of limited interest for practical purposes \cite{Stynes,Hanyga}.  

\section{Connection with General Fractional Calculus} \label{sec:gfc}

	Let us first recall the rigorous definition of {\em Sonine pairs} \cite{samko2003integral}. Two functions $k(t) , \, \varkappa (t) \in L_{\rm loc} ^1 (\R^+)$ form a Sonine pair if $(k \star \varkappa) (t) = 1$ for almost all $t>0$. If one then restricts the operators in Section \ref{sec:kebab} to the case of 
\be \label{restrictions-gfc}
\alpha, \beta, \lambda, \gamma \in \mathbb{\R} \, , \quad \alpha>0 \, , \quad 0<\beta<1 \, ,
\ee
it is easy to see that 
\be \label{restrictions-kernel}
k(t) = e^{- \gamma} _{\alpha, 1-\beta}(\lambda; t) \, , \quad \varkappa(t) = e^{\gamma}_{\alpha,\beta}(\lambda ; t) \, ,
\ee
form a Sonine pair. Indeed, denoting by
$$ \wt{f}(s) \equiv \mathcal{L} \left[ f(t) \, ; \, s \right] = \int _0 ^\infty e^{-st} \,  f(t) \, \d t $$
the Laplace transform of a function $f(t)$, and recalling that
$$ \mathcal{L} \left[e^{\gamma} _{\alpha, \beta}(\lambda; t) \, ; \, s \right] = s^{-\beta} (1 - \lambda s^{-\alpha})^{-\gamma} \, , $$
one can easily infer that 
$$ \wt{\varkappa} (s) = \frac{1}{s \, \wt{k} (s)} \, , $$
thus verifying that the two functions in \eqref{restrictions-kernel} form a Sonine pair.

	Now, pursuing the general idea in \cite{Kochubei}, it is important to consider the role of relaxation processes in order to understand to what extent this picture can be reconciled with the scheme of general fractional calculus. Specifically, let us assume $k(t)$, appearing in Eq. \eqref{C-K}, to have a well define Laplace transform $\wt{k} (s)$ for all $s>0$, such that $\wt{k}(s)$ is a function of the Stieltjes class and satisfies: (i) $\wt{k}(s) \to \infty$ as $s \to 0$; (ii) $s \,\wt{k}(s) \to 0$ as $s \to 0$; (iii) $\wt{k}(s) \to 0$ as $s \to \infty$; (iv) $s \,\wt{k}(s) \to \infty$ as $s \to \infty$. Then, the Cauchy problem
\be
\mathcal{D} _{(k)} y(t) = - \lambda \, y(t) \, , \quad y(0) = y_0 \, , \qquad \lambda > 0 \, , \,\, t>0 \, ,
\ee
known as the {\em relaxation problem}, admits a unique solution which turns out to be both infinitely differentiable and completely monotonic.

Coming to the Prabhakar derivative, considering the restrictions in \eqref{restrictions-gfc}, one has that Eq. \eqref{P-derivative} reduces to (see \cite{PBK-derivative})
\be \label{eq:P-gfc}
{}^{\rm C}\bm{\mathcal{D}}^{\gamma}_{\alpha, \beta, \lambda} f(t)=
\bm{\mathcal{E}}^{-\gamma}_{\alpha, 1-\beta, \lambda}  f' (t) =
\frac{\d}{\d t} \int _0 ^t k(t - \tau) \, f(\tau) \, \d \tau 
- k(t) \, f(0^+)
\ee
with $k(t) = e^{- \gamma} _{\alpha, 1-\beta}(\lambda; t)$, as in \eqref{restrictions-kernel}. Since the derivative in Eq. \eqref{eq:P-gfc} has a form compatible to the one in \cite{Kochubei}, one can now wonder about the nature of the solutions of 
the Cauchy problem
\be \label{cauchy-P}
{}^{\rm C}\bm{\mathcal{D}}^{\gamma}_{\alpha, \beta, \lambda} y(t) = - \xi \, y(t) \, , 
\quad y(0) = y_0 \, , \qquad \xi > 0 \, , \,\, t>0 \, .
\ee
Note that the explicit solution of this problem, obtained through the Laplace transform method, is discussed in \cite{GC}, and further analyzed in \cite{ZhaoSun, ZhaoSun-note}.

Focusing on the kernel function $k(t)$ in Eq. \eqref{restrictions-kernel}, it is easy to see that its Laplace transform, \ie
\be \label{kernel-P-k} 
\wt{k} (s) = s^{\beta-1} (1 - \lambda s^{-\alpha})^{\gamma} \, , 
\ee
satisfies the conditions (i)-(iv) if 
\be \label{ineq-1}
0<\beta<1 \, , \qquad - \alpha \gamma < 1 -\beta < 1 - \alpha \gamma \, ,
\ee
however, showing in which settings \eqref{kernel-P-k} is a Stieltjes function turns out to be a rather hard task in general.
Alternatively, one could try to deduce for which values of the parameters $(\alpha, \beta, \gamma, \lambda)$ the function \eqref{kernel-P-k} turns out to be of the Stieltjes class, by means of a reverse engineering approach, even though this procedure will not lead to a general result on the connection between Prabhakar's theory and Kochubei's general fractional calculus. Hence, recalling that a Stieltjes function can be thought of as the Laplace transform of a completely monotonic function \cite{Stieltjes}, then one can take profit of the results in \cite{MG} to provide some constraints on the parameters in \eqref{restrictions-gfc} so that the solution of \eqref{cauchy-P} turns out to be completely monotonic.\footnote{It is worth mentioning that the complete monotonicity of the Prabhakar function has originally been investigated in \cite{hanyga-jst}, within the framework of anisotropic dielectric relaxation.} In detail, recalling that the function $e^{\eta} _{\nu, \sigma}(-\mu; t)$, with $\mu >0$, is locally integrable and completely monotonic if (see \eg \cite{MG})
\be
0 < \nu \leq 1 \, , \quad 0 < \nu \eta \leq \sigma \leq 1 \, , 
\ee
then $\wt{k} (s)$ in \eqref{kernel-P-k} will surely be a Stieltjes function if  
\be \label{ineq-2}
\lambda < 0 \, , \quad 0 < \alpha \leq 1 \, , \quad 0 < - \alpha \gamma \leq 1 -\beta \leq 1 \, .
\ee
Putting together \eqref{ineq-1} and \eqref{ineq-2}, one finds
\be \label{final-restrictions}
\lambda < 0 \, , \quad 
0 < \alpha \leq 1 \, , 
\quad \gamma < 0 \, , \quad 0 < \beta < 1 \, , \quad - \alpha \gamma \leq 1 -\beta \leq 1 \, .
\ee
Hence, considering the Prabhakar derivative \eqref{eq:P-gfc} with parameters as in \eqref{final-restrictions}, then the Cauchy problem \eqref{cauchy-P} admits a unique infinitely differentiable and completely monotonic solution. Note that, some results concerning the connection of Prabhakar's theory with Kochubei's general fractional calculus have been analyzed in \cite{ZhaoSun-note} for the specific set of parameters employed in \cite{ZhaoSun}.

	To sum up, one can conclude that, for some values of the parameters, Prabhakar's theory can be related to Kochubei's general fractional calculus. However, this does not represent a general result for the full theory since it is not guaranteed that all the allowed configurations of the parameters in \eqref{restrictions-gfc} will fit this general scheme.

\section{Conclusions}\label{sec:conclusions}

After reviewing the general ideas of Kochubei's general fractional calculus in light of some recent results \cite{Hanyga} concerning integro-differential operators with weakly singular kernels, I have discussed to what extent Prabhakar calculus can be framed within this general scheme. To this aim I have also investigated the behavior of the Prabhakar kernel \eqref{kernel-derivative}. It turns out, as one would probably expect from \cite{SC-1,SC-2,Hanyga}, that if $\Re (\beta) >0$ and $\beta \notin \N$ then \eqref{kernel-derivative} is always weakly singular. On the other hand, if we take the integer limit for $\beta$, the Prabhakar kernel \eqref{kernel-derivative} leaves $L ^1 _{\rm loc} (\R ^+)$, entering the realm of distributions. This implies that, in this limit, the Prabhakar derivative can be split into an ordinary derivative and a series of linear Volterra-like integro-differential operators, with either regular ($\Re(\alpha) \geq 1$) or (at least in part) singular ($0<\Re(\alpha) < 1$) kernels. This result, in all honesty, makes this limit of minor interest for fractional calculus. Finally, I have discussed some generalities of relaxation processes involving Prabhakar derivatives taking profit of the tools and general arguments in \cite{Kochubei, MG} and building further on the results in \cite{ZhaoSun-note}.

\section*{Acknowledgments}	
The author is grateful to Y. Luchko and A. Hanyga for useful comments. This work is supported, in part, by the Natural Science and Engineering Research Council of Canada (Grant No. 2016-03803 to V. Faraoni) and by Bishop's University. Furthermore, this study was carried out in the framework of the activities of the National Group of Mathematical Physics (GNFM, INdAM).

\appendix

\section{On the Prabhakar derivative and its regularization} \label{sec:appendix}
As mentioned in Section \ref{sec:kebab}, in \cite{Prabhakar1971} T. Prabhakar introduced the notion of three-parameter Mittag-Leffler function, often referred to as Prabhakar function, as a kernel for a class of singular integral equations. These equations are defined through an integral operator which takes the form in \eqref{P-integral}. Such an operator was later named after the author and is currently known as the Prabhakar fractional integral. An extensive study of this operator was then carried out by many authors (see \eg \cite{Giusti-NODY, PBK-derivative, GG, MG, GC, Garrappa-HN, CGV-Mathematics2018, TomovskiHilferSrivastava2010,KSS}) several decades after its original formulation. 

The first proposal of a Riemann-Liouville type definition of fractional derivative based on the Prabhakar fractional integral was formulated by A. A. Kilbas, M. Saigo, and R. K. Saxena in \cite{KSS} and reads
\be \label{P-derivative-RL}
\bm{\mathcal{D}}^{\gamma}_{\alpha, \beta, \lambda} f(t)&=D^m \bm{\mathcal{E}}^{-\gamma}_{\alpha, m-\beta, \lambda} f (t) \, , \qquad m:= \ceil{\beta} \, ,
\ee 
with all the appropriate assumptions on the regularity of $f(t)$ (see \cite{KSS,PBK-derivative}) and with $(\alpha,\beta,\gamma, \lambda)$ as in Section \ref{sec:kebab}. Later on, M. D'Ovidio and F. Polito named this operator after T. Prabhakar in \cite{DOP}, giving birth to the (well-established) terminology used in the current literature, and introduced the notion of regularized Prabhakar derivative, defined as in \eqref{P-derivative}. These notions were then collected and further built upon by R. Garra, {\em et al.} in \cite{PBK-derivative}, kick-starting the current research line on Prabhakar's theory.

The Prabhakar derivative \eqref{P-derivative-RL}, introduced by A. A. Kilbas, M. Saigo, and R. K. Saxena, was therefore the first proposal for a left-inverse operator of the Prabhakar fractional integral \eqref{P-integral}. It is now worth pointing out that the regularized Prabhakar derivative \eqref{P-derivative} does act as a left-inverse of the Prabhakar fractional integral \eqref{P-integral}. To prove this statement I need to recall a few important properties of both the Prabhakar fractional integral and derivatives.

First, let $\alpha, \beta, \lambda, \gamma \in \mathbb{R}$ and $\alpha, \beta>0$. Furthermore, let $f(t)$ be a real-valued function whose derivatives are continuous up to order $m-1$ on $[0,b]$ and such that $D^{m-1} f$ is absolutely continuous on $[0,b]$, with $0 < t < b \leq \infty$. Then,
\be \label{cond-1}
{}^{\rm C}\bm{\mathcal{D}}^{\gamma}_{\alpha, \beta, \lambda} f(t) =
\bm{\mathcal{D}}^{\gamma}_{\alpha, \beta, \lambda} \Bigg[ f(t) - \sum _{k=0} ^{m-1} \frac{t^k}{k!} \, f^{(k)}(0^+)
 \Bigg] 
\, ,
\ee
see \cite{PBK-derivative} for details.

Secondly, assume $\alpha, \beta, \lambda, \gamma \in \mathbb{R}$, with $\alpha, \beta>0$, and $\beta > k$ with $k \in \N$. If $f \in C[0,b]$, then
\be \label{cond-2}
D^k \bm{\mathcal{E}}^{\gamma}_{\alpha,\beta, \lambda} f(t) =
\bm{\mathcal{E}}^{\gamma}_{\alpha,\beta - k, \lambda} f(t) \, ,
\ee
see \cite{KSS} for details.

One can then prove the following 
\begin{thm}
Let $\alpha, \beta, \lambda, \gamma \in \mathbb{R}$ and $\alpha, \beta>0$. If $f \in C[0,b]$, then the regularized Prabhakar derivative \eqref{P-derivative} is a left-inverse of the Prabhakar fractional integral \eqref{P-integral}.
\end{thm}

\begin{proof}
From \eqref{cond-1} one finds
$$
{}^{\rm C}\bm{\mathcal{D}}^{\gamma}_{\alpha, \beta, \lambda} \bm{\mathcal{E}}^{\gamma}_{\alpha,\beta, \lambda} f(t) =
\bm{\mathcal{D}}^{\gamma}_{\alpha, \beta, \lambda} \Bigg[ 
\bm{\mathcal{E}}^{\gamma}_{\alpha,\beta, \lambda} f(t)
- \sum _{k=0} ^{m-1} \frac{t^k}{k!} \, \big(D^k \bm{\mathcal{E}}^{\gamma}_{\alpha,\beta, \lambda} f\big)(0^+)
 \Bigg]\, .
$$
Because of \eqref{cond-2} one also has that \cite{KSS}
$$
\big(D^k \bm{\mathcal{E}}^{\gamma}_{\alpha,\beta, \lambda} f\big)(0^+) = 
\big(\bm{\mathcal{E}}^{\gamma}_{\alpha,\beta - k, \lambda} f\big)(0^+) \, .
$$
Since $f(t)$ is continuous, for $t \in B_{\delta} \subset [0,b]$, with $B_{\delta} = (0, \delta)$ and $\delta>0$, one finds
\begin{eqnarray*}
0 \leq
\Big| \bm{\mathcal{E}}^{\gamma}_{\alpha,\beta - k, \lambda} f (t) \Big| &\!\!=\!\!& 
\Bigg| \int_0 ^t (t-\tau)^{\beta - k - 1} E^{\gamma} _{\alpha, \beta - k} \Big( \lambda (t - \tau)^\alpha \Big) \, f(\tau) \, \d \tau \Bigg| \\ 
&\!\!\leq\!\!& {\rm ess \, sup}_{t \in B_{\delta}} \big| f (t) \big| \, 
t^{\beta - k} \, E^{\gamma} _{\alpha, \beta - k + 1} ( \lambda t^\alpha ) \\
&\overset{t \to 0^+}{\longrightarrow}& 0 \, ,
\end{eqnarray*}
taking profit of the property
$$
\int _0 ^x t^{\beta - 1} \, E^{\gamma} _{\alpha, \beta} ( \lambda t^\alpha ) \, \d t = 
x^{\beta} \, E^{\gamma} _{\alpha, \beta + 1} ( \lambda x^\alpha ) \, ,
$$
with $\alpha, \beta, \gamma, \lambda \in \C$ and $\Re(\alpha) , \Re(\beta) > 0$, see \eg \cite{KSS}.

Then,
\begin{eqnarray*}
{}^{\rm C}\bm{\mathcal{D}}^{\gamma}_{\alpha, \beta, \lambda} \bm{\mathcal{E}}^{\gamma}_{\alpha,\beta, \lambda} f(t) &\!\!=\!\!&
\bm{\mathcal{D}}^{\gamma}_{\alpha, \beta, \lambda} \Bigg[ 
\bm{\mathcal{E}}^{\gamma}_{\alpha,\beta, \lambda} f(t)
- \sum _{k=0} ^{m-1} \frac{t^k}{k!} \, \big(D^k \bm{\mathcal{E}}^{\gamma}_{\alpha,\beta, \lambda} f\big)(0^+)
 \Bigg]\\
&\!\!=\!\!&
\bm{\mathcal{D}}^{\gamma}_{\alpha, \beta, \lambda} \Bigg[ 
\bm{\mathcal{E}}^{\gamma}_{\alpha,\beta, \lambda} f(t)
- \sum _{k=0} ^{m-1} \frac{t^k}{k!} \, \big(\bm{\mathcal{E}}^{\gamma}_{\alpha,\beta-k, \lambda} f\big)(0^+)
 \Bigg]\\
&\!\!=\!\!&
\bm{\mathcal{D}}^{\gamma}_{\alpha, \beta, \lambda} 
\bm{\mathcal{E}}^{\gamma}_{\alpha,\beta, \lambda} f(t) = f(t) \, ,
\end{eqnarray*}
where in the last step I have used the fact that the Prabhakar derivative \eqref{P-derivative-RL} is the left-inverse of the Prabhakar integral \eqref{P-integral} (see \cite{KSS}), thus concluding the proof.
\end{proof}

%
%
%
%
%
%
%
%
%

\end{document}